\documentclass[11pt]{article}
\usepackage{fullpage,times}

\usepackage[utf8]{inputenc}
\usepackage[T1]{fontenc}

\usepackage{amsmath}
\usepackage{amssymb}
\usepackage{bbm}
\usepackage{amstext}
\usepackage{amsfonts}
\usepackage{amsthm}
\usepackage{mathtools}
\usepackage{scrextend}    
\usepackage{array}     
\usepackage{multicol}    
\usepackage{mathrsfs}    
\usepackage{color}
\usepackage{hyperref}
\usepackage{cleveref}
\usepackage{bm}
\usepackage{enumitem}
\usepackage{upgreek}
\usepackage{mleftright}
\usepackage[mathscr]{euscript}
\usepackage{thm-restate}
\usepackage[linesnumbered,ruled,vlined]{algorithm2e}

\newtheorem{theorem}{Theorem}[section]    
\newtheorem{definition}{Definition}[section] 
\newtheorem{corollary}[theorem]{Corollary}    
\newtheorem{lemma}[theorem]{Lemma}    

\renewcommand{\qed}{\hfill{$\rule{6pt}{6pt}$}} 
\renewenvironment{proof}{\noindent{\bf Proof:}}{\qed\\}

\numberwithin{equation}{section}

\newcommand{\complex}{{\mathbb C}}

\newcommand{\integers}{{\mathbb Z}}
\newcommand{\naturals}{{\mathbb N}}

\newcommand{\norm}[1]{\left\| #1 \right\|}
\newcommand{\trnorm}[1]{\left\| #1 \right\|_{\mathrm{tr}}}
\newcommand{\size}[1]{\left| #1 \right|}
\newcommand{\set}[1]{\left\{ #1 \right\}}

\newcommand{\ceil}[1]{{\lceil #1 \rceil}}

\DeclareMathOperator{\trace}{Tr}
\DeclareMathOperator{\support}{supp}
\DeclareMathOperator{\range}{range}
\DeclareMathOperator{\Order}{O}
\DeclareMathOperator{\order}{o}

\DeclareMathOperator{\expct}{{\mathbb E}}

\DeclareMathOperator{\e}{e}
\newcommand{\complexi}{{\mathrm{i}}}

\newcommand{\id}{{\mathbb I}}

\newcommand{\ket}[1]{| #1 \rangle}

\newcommand{\ketbra}[2]{| #1 \rangle\!\langle #2 |}
\newcommand{\braket}[2]{\langle #1 | #2 \rangle }
\newcommand{\density}[1]{\ketbra{#1}{#1}}

\newcommand{\eqdef}{\coloneqq}
\newcommand{\tensor}{\otimes}

\newcommand{\union}{\cup}
\newcommand{\intersect}{\cap}

\newcommand{\adjoint}{*}
\newcommand{\setcomp}[1]{\overline{#1}}

\newcommand{\suppress}[1]{}

\newcommand{\etal}{\emph{et al.\/}}

\newcommand{\speqdef}{\quad \eqdef \quad}
\newcommand{\speq}{\quad = \quad}
\newcommand{\sple}{\quad \le \quad}
\newcommand{\spge}{\quad \ge \quad}

\newcommand{\cA}{\mathcal{A}}
\newcommand{\cB}{\mathcal{B}}

\newcommand{\cT}{\mathcal {T}}

\DeclareMathOperator{\VC}{VC\mbox{-}dim}
\DeclareMathOperator{\QPEX}{QPEX}

\newcommand{\scC}{{\mathscr C}}

\newcommand{\bi}{{\bm i}}
\newcommand{\bI}{{\bm I}}
\newcommand{\bj}{{\bm j}}
\newcommand{\bx}{{\bm x}}
\newcommand{\bX}{{\bm X}}
\newcommand{\by}{{\bm y}}
\newcommand{\bb}{{\bm b}}

\newcommand{\N}{\naturals}
\newcommand{\ms}{\mathrm{ms}}
\newcommand{\supp}{\mathrm{supp}}

\DeclareMathOperator{\opt}{opt}

\title{\textbf{Proper versus Improper Quantum PAC Learning}}

\author{
Ashwin Nayak~\thanks{Department of Combinatorics and Optimization,
and Institute for Quantum Computing, University
of Waterloo, 200 University Ave.\ W., Waterloo, ON,
N2L~3G1, Canada.
Email: \texttt{academic@ashwinnayak.info}~.
} \\
University of Waterloo
\and
Pulkit Sinha~\thanks{School of Computer Science,
and Institute for Quantum Computing, University
of Waterloo, 200 University Ave.\ W., Waterloo, ON,
N2L~3G1, Canada.
Email: \texttt{psinha@uwaterloo.ca}~.
} \\
University of Waterloo}

\date{February 22, 2024}

\begin{document}

\maketitle

\begin{abstract}
A basic question in the PAC model of learning is whether \emph{proper\/} learning is harder than \emph{improper} learning. In the classical case, the answer to this question, with respect to sample complexity, is known to depend on the concept class. While there are concept classes for which the two modes of learning have the same complexity, there are examples of concept classes with VC dimension~$d$ that have sample complexity~$\Omega \left( \frac{d}{\epsilon}\log \frac 1 \epsilon \right)$ for proper learning with error~$\epsilon$, while the complexity for improper learning is~$\Order \!\left( \frac{d}{\epsilon}\right) $. One such example arises from the Coupon Collector problem.

Motivated by the efficiency of proper versus improper learning with \emph{quantum\/} samples, Arunachalam, Belovs, Childs, Kothari, Rosmanis, and de Wolf~\cite{ABCKRW20-quantum-coupon-collector} studied an analogue, the Quantum Coupon Collector problem. Curiously, they discovered that for learning size~$k$ subsets of~$[n]$ the problem has sample complexity $\Theta(k \log \min \set{k, n-k + 1})$, in contrast with the complexity of~$\Theta(k \log k)$ for Coupon Collector. This effectively negates the possibility of a separation between the two modes of learning via the quantum problem, and Arunachalam \etal\ posed the possibility of such a separation as an open question.

In this work, we first present an algorithm for the Quantum Coupon Collector problem with sample complexity that matches the sharper lower bound of~$ (1 - \order_k(1)) k \ln \min \set{k, n-k+1}$ shown recently by Bab Hadiashar, Nayak, and Sinha~\cite{BNS24-sample-complexity}, for the entire range of the parameter~$k$. Next, we devise a variant of the problem, the Quantum \emph{Padded\/} Coupon Collector. We prove that its sample complexity matches that of the classical Coupon Collector problem for both modes of learning, thereby exhibiting the same asymptotic separation between proper and improper quantum learning as mentioned above. The techniques we develop in the process can be directly applied to any form of padded quantum data. We hope that padding can more generally lift other forms of classical learning behaviour to the quantum setting.
\end{abstract}

\section{Introduction}
\label{sec-introduction}

One of the most fundamental tasks in machine learning is to predict certain properties of objects, given access to labeled data for other objects drawn from the same distribution. Problems like these form the basis of Supervised Machine Learning. Setting aside the algorithmic aspects of the problem, a natural question to ask is how many labeled data are required to achieve a certain degree of accuracy in the prediction?

To study this formally, Valiant~\cite{val84-PAC} developed the \textit{probably approximately correct\/} (PAC) model for such learning tasks. In this model, binary properties are thought of boolean functions, mapping each object in the population to a 0 or 1 corresponding to whether they have the property or not. Depending on the learning task at hand, we may have some prior knowledge about the boolean functions that may be allowed. Each boolean function is called a \textit{concept}, and the set of allowed functions, a \textit{concept class}. Given labeled data according to an unknown concept in the concept class and an unknown distribution on the population, the goal is to reconstruct a concept (with probability at least~$1-\delta$) which on the same distribution matches with the unknown concept with probability $(1-\epsilon)$. The data needed is formalised as the sample complexity, which is simply the number of labeled samples. This is called the $(\epsilon,\delta)$-PAC sample complexity of the learning task. 

Formalised this way, the hardness of the learning task only really depends on the concept class. Turns out, as shown in a series of works~\cite{BEHW89-Learn-VC,Han16-opt-PAC}, this hardness is almost fully characterized by the \textit{VC dimension} $d$ of the concept class, and the sample complexity has been shown to be 
\begin{equation}
\label{eq:PAC-sample-complexity_intro}
    \Theta \left( \frac{d}{\epsilon} + \frac{\log (1/\delta)}{\epsilon}\right) \enspace.
\end{equation}

This learning model was extended to the quantum setting by Bshouty and Jackson~\cite{BJ99-DNF}, where random samples of labeled data were replaced with \textit{quantum\/} samples. In quantum samples, the data points are in a superposition, weighted with the square root of the corresponding probability. Learning using these samples is readily seen to be at least as easy as with classical samples, as we can apply classical learning techniques on the distribution obtained after measuring in the standard basis. Hence the upper bound on sample complexity in Eq.~\eqref{eq:PAC-sample-complexity_intro} applies. Arunachalam and de Wolf~\cite{AW18-optimal-sample} showed that the lower bound in \eqref{eq:PAC-sample-complexity_intro} also holds in the quantum setting. (See also Ref.~\cite{BNS24-sample-complexity} for an arguably simpler, information-theoretic proof.)

Interestingly, optimal PAC learners are not guaranteed to output a concept within the concept class. The optimal learner, even though equipped with information about the possible actual classifications, can choose to not classify based on any of those. This kind of learning, where the output concept is not guaranteed to belong to the concept class, is called \textit{improper learning}. If on the other hand, we enforce the condition that the PAC learner only output concepts from within the class, then the learning task is called \textit{proper learning}. 

Does the requirement of proper learning make the task harder? In the classical case, the answer to this question, with respect to sample complexity, is known to depend on the concept class. For the concept class containing all possible concepts, the answer is a clear no, as proper and improper learning are equivalent. On the other hand, for each $d$ and $\epsilon$, there are examples of concept classes that have sample complexity
\begin{equation}
\label{eq-proper-cc-bd}
\Omega \left( \frac{d}{\epsilon}\log \frac 1 \epsilon + \frac{\log (1/\delta)}{\epsilon}\right)
\end{equation} 
for proper learning.
These classes exhibit an asymptotic separation between the sample complexity of proper and improper PAC learning. An apparently folklore example of such concept classes arises from the \emph{Coupon Collector\/} problem~\cite{H19-proper-vs-improper}.

In the Coupon Collector problem, a classic problem from probability theory, we are given integers~$k,n$ with~$1 < k < n$, and have access to independent, uniformly distributed elements from a fixed, unknown size-$k$ subset~$S$ of~$[n]$. The task is to draw enough samples so as to ``collect'' (i.e., observe) all~$k$ ``coupons'' (elements) in~$S$. Viewed as a learning problem, the task is to determine the unknown set~$S$. It is well-known that~$k \ln k + \Theta(k)$ samples are necessary and sufficient to observe all~$k$ elements of~$S$ with constant probability of success (see, e.g.,~\cite[Theorem~5.13]{MU17-probability-computing}). 

Taking~$n \eqdef \ceil{1/\epsilon}$ and~$k \eqdef n - 1$ in the Coupon Collector problem, we get a concept class with VC-dimension~$1$, and hence the class is~$(1/n, 1/4)$-PAC learnable with~$\Theta(n)$ samples, albeit \emph{improperly\/}. On the other hand, \emph{proper\/} PAC learning the class with the same parameters requires~$\Theta(n \ln n)$ samples, as this entails correctly producing the size-$(n - 1)$ subset with probability at least~$3/4$. This gives us a separation of~$\Omega( (1/\epsilon) \log (1/\epsilon))$ versus~$\Order(1/\epsilon)$ for the sample complexity of proper versus improper learning. A generalization of this construction gives us a concept class with VC-dimension~$d$ which has sample complexity stated in Eq.~\eqref{eq-proper-cc-bd}~\cite{H19-proper-vs-improper}.

Motivated by the efficiency of proper versus improper learning with \emph{quantum\/} samples, Arunachalam, Belovs, Childs, Kothari, Rosmanis, and de Wolf~\cite{ABCKRW20-quantum-coupon-collector} studied a quantum analogue of the Coupon Collector problem. In the \emph{Quantum Coupon Collector\/} problem, the learner has access to quantum samples, i.e., uniform superpositions~$\ket{\psi_S}$ over the elements of an unknown size-$k$ subset~$S$ of~$[n]$. The task is to learn~$S$ with high probability. Arunachalam \etal\ discovered a curious phenomenon---the quantum sample complexity of the problem matches that of classical coupon collection for~$k$ up to roughly~$n/2$, but then \emph{decreases\/} until~$k$ reaches~$n$. More precisely, with~$m \eqdef n - k$, they proved that the sample complexity of the Quantum Coupon Collector problem with constant probability of error~$< 1$ is~$\Theta(k \log \min \set{k, m + 1})$. Bab Hadiashar, Nayak, and Sinha~\cite{BNS24-sample-complexity} proved a sharper lower bound of~$(1 - \order_k(1)) k \ln \min \set{k, m+1}$. More precisely, they established the following result.
\begin{theorem}[\cite{BNS24-sample-complexity}, Theorem~IV.17 and proof of Corollary~IV.18]
\label{thm-qcc-lb}
Let~$\delta\in (0,1/40]$ and~$c_0 \eqdef \tfrac{1}{2} \ln \big(\tfrac{1 - \delta}{32 \delta } \big)$. Any algorithm for the Quantum Coupon Collector problem with parameters~$n,k$, and error probability at most~$\delta$ has sample complexity
\begin{itemize}
\item at least~$k \ln m + c_0 n$ when~$1 \le m \le \delta n$ and~$m\ln m \le c_0 n/20$, and
\item at least~$k \ln k - k \ln\ln k - \Order(k)$ otherwise.
\end{itemize}
\end{theorem}
The above bound has the exact optimal leading order term for~$k$ sufficiently smaller than~$n$ (listed as the second case in the theorem). Bab Hadiashar \etal\ conjectured that this sharp optimality of the bound extends to all~$k$ up to~$n$. 

First, we present an algorithm that confirms this conjecture. 
\begin{theorem}
\label{thm-qcc-optimal}
There is an algorithm for the Quantum Coupon Collector problem with parameters~$n,k$, where~$n \ge 1$ and~$1 < k < n$, that has probability of error at most~$\delta \in (0,1]$ and uses 
\begin{itemize}
\item at most~$k \ln m + k \ln \frac{\e}{\delta}$ samples when~$3m \ln( \e m) \le n \,$, and
\item at most~$k \ln k + k \ln \frac{1}{\delta} $ samples otherwise.
\end{itemize}
\end{theorem}
We may verify that the number of samples matches the lower bound in \Cref{thm-qcc-lb} exactly in the leading order term for all~$k$, for constant error~$\delta$. Moreover, the second order term is also tight up to a factor of~$2$, as a function of~$k$ and~$\delta$, in the first case listed in \Cref{thm-qcc-lb}.

When~$k$ is sufficiently small, i.e., in the ``large''~$m$ regime, there is a straightforward algorithm. We measure the quantum samples and to obtain classical ones, and follow the classical learning scheme. The non-trivial case is when~$m$ is ``small''. In the small~$m$ regime, the previous best algorithm due to Arunachalam \etal\ attempts to learn the complement~$\setcomp{S}$ rather than~$S$. The algorithm repeatedly measures the quantum samples~$\ket{\psi_S}$ according to the measurement~$(\density{\psi_{[n]}}, \id - \density{\psi_{[n]}})$, where~$\ket{\psi_{[n]}}$ is the uniform superposition over all elements of~$[n]$. The algorithm further measures the residual state in the computational basis when the second outcome is observed. As we may readily verify, the elements in~$\setcomp{S}$ are considerably more likely to be observed. Thus it suffices to estimate the frequency with which elements occur to determine which ones belong to~$\setcomp{S}$.

We begin with the basic idea underlying the above algorithm, and build on it. Instead of estimating frequencies, we take a more optimistic approach, in that we gamble on the element observed in the second measurement as being one of the more likely elements. I.e., we gamble on the element belonging to~$\setcomp{S}$. Of course, the gamble may fail, and we may incorrectly add an element to our guess for~$\setcomp{S}$. We devise a refined set of measurements to either identify and remove such ``rogue'' elements from our guess, or to discover other elements that are likely to be in~$\setcomp{S}$. Perhaps surprisingly, the behaviour of the resulting random process resembles that of the classical Coupon Collector process, and we show that it converges sufficiently rapidly to the correct set. 

Returning to the question of proper versus improper learning, consider the case~$k = n - 1$. The sample complexity of the Quantum Coupon Collector problem is~$\Theta(n)$. Thus, unlike in the classical case, the quantum version of the problem \emph{does not\/} separate the sample complexity of proper from that of improper quantum PAC learning. Arunachalam \etal\ concluded their work by raising the question of the asymptotic equality of the complexity of learning in the two modes~\cite[Section~5]{ABCKRW20-quantum-coupon-collector}.

The second result we present resolves this question.
\begin{theorem}
\label{thm-qpac-prop-vs-improp}
For each $d\geq 5$ and $\epsilon <1$, there are concept classes for which $\Omega\left(\frac d \epsilon \log \frac 1 \epsilon+\frac 1 \epsilon \log \frac 1 \delta \right)$ quantum samples are required for proper $(\epsilon,\delta)$-PAC learning, while $\Order \! \left(\frac d \epsilon +\frac 1 \epsilon \log \frac 1 \delta \right)$ quantum samples suffice for improper $(\epsilon,\delta)$-PAC learning.
\end{theorem}
To establish this separation, we construct concept classes that are ``padded'' versions of those used to show the separation via the classical Coupon Collector problem. The \textit{Padded Coupon Collector\/} problem is obtained by modifying the standard Coupon Collector problem by appending arbitrary values from a large set to each element of the set $[n]$. Its quantum version, the \textit{Quantum Padded Coupon Collector\/} (QPCC) is defined analogously. The goal in both the cases is still to learn an unknown subset of~$[n]$ as in the standard coupon collector problem, or more generally, to learn a large part of the subset. (See \Cref{sec-preliminaries} for the precise definition of the problems.) In the classical case, the addition of padding has no real consequence. As the padding for any element does not carry any information about the set, an optimal learner always ignores it. In the quantum case, however, the situation becomes interesting. 

In contrast with the Quantum Coupon Collector problem, we show that the behaviour of the quantum samples for this new problem, with respect to learning, is very close to that in the classical case. The techniques we develop approximately preserve not only the classical sample complexity of learning the exact set, but also the for the case when we allow a small number~$l$ of \textit{mismatches}---elements that are in the hypothesis, but not in the underlying set.
\begin{restatable}{theorem}{thmqpccbound}
\label{QPCC_thm}
Consider the Quantum Padded Coupon Collector problem for size~$k$ subsets of~$[n]$ with at most~$l$ mismatches and pads ranging in~$\naturals_p \,$.
Any learner for this problem that has success probability at least~$1-\delta$ uses at least
\[t_0 \speqdef k \ln \frac {k+1}{10l+1} + k \ln (1-4\delta)\]
samples, provided~$ n \geq k+5l$, $ l \geq 1$, $p \geq  k^{t_0} / \delta^2 $, and~$\delta \in (0, 1/4)$.
\end{restatable}
Note that~$t_0 \in \Omega(k\log \frac{k}{l} )$ for constant~$\delta$.
\Cref{QPCC_thm} helps us establish not only the $\Omega\left(\frac 1 \epsilon \log \frac 1 \epsilon  \right)$ proper PAC learning lower bound for $d \in \Theta(1)$ known for classical learning, but also the more general lower bound of $\Omega\left(\frac d \epsilon \log \frac 1 \epsilon  \right)$ stated in \Cref{thm-qpac-prop-vs-improp}.

To understand the role of padding, it is instructive to examine how sample-efficient learning algorithms for Quantum Coupon Collector work. The algorithm in Ref.~\cite{ABCKRW20-quantum-coupon-collector} and the new algorithm we present in Section~\ref{sec-qcc-algorithm} both heavily exploit the property that when~$k$ is ``close'' to~$n$ the quantum samples corresponding to different subsets are very close to the uniform superposition of all elements. This helps us convert quantum samples for the sets to those of their \emph{complements\/}, with high probability. Padding enables us to circumvent such algorithms. Since the padding of each element in $[n]$ is an unknown arbitrary value, the quantum samples corresponding to the different subsets with the padding are no longer close to a common pure state, immediately thwarting the known algorithms. Intuitively, due to the independence of the padding from the sets, if the number of possible ways of padding is very large, one may expect that any attempt by a learner to incorporate the register holding the padding in its operations will disturb the coherence of the quantum sample. On the other hand, not using the padding at all is equivalent to tracing it out, again reducing quantum samples to classical samples. 

The techniques underlying the proof of \Cref{QPCC_thm}, however, formalize this intuition only indirectly. They have their origins in the spectral analysis that lies at the heart the optimal lower bounds for PAC and agnostic quantum learning due to Bab Hadiashar, Nayak, and Sinha~\cite{BNS24-sample-complexity}. We consider the learning problem for sets and the padding chosen uniformly. Taking advantage of the symmetry introduced by the padding, we show that the quantum state corresponding to the random quantum samples may be viewed as a mixture of states that do not contain sufficient information about the set if the number of samples is not sufficiently large. In fact, we relate the information in the states to that obtained from the same number of classical samples. Known properties of the Coupon Collector process then lead to the theorem.

The techniques we develop for QPCC can be applied to padded quantum data in other contexts as well. So we hope that padding can more generally lift other forms of classical learning behaviour to the quantum setting. In particular, these ideas may help answer some questions in Quantum Learning Theory raised in the recent survey by Anshu and Arunachalam~\cite{AA24-learning-survey}.

On a more philosophical note, at first glance, the idea of padding might seem artificial. However, in the setting of classical learning, in many cases padding gives a more accurate picture of the kind of data that are available, and of the corresponding learning process. Often, for simplicity and/or computational efficiency, while trying to learn aspects of a population via sampling, we ignore parts of the data from each sample. As a simple example, consider estimating the average height of individuals in a population. For this task, most if not all sensible estimators ignore features such as weight and eye-color, even if such information is present in the samples. This extra unused information behaves like padding.

\paragraph{Organization of the paper.} In Section \ref{sec-preliminaries} we describe the notation for properties of sequences, as well as for notions in quantum information. This is followed by basics of Quantum Learning Theory and the PAC learning model in both the classical and quantum settings. We then describe the relevant variants of the Coupon Collector problem: the Quantum Coupon Collector, the Padded Coupon Collector, and the Quantum Padded Coupon Collector. At the end we state a tail bound relevant to the properties of the Classical Coupon Collector problem.

Section \ref{sec-coupon-collector} contains the relevant arguments for proving Theorem \ref{thm-qcc-optimal}; Section \ref{sec-qcc-algorithm} contains the details of the algorithm while Section \ref{sec-alg-analysis} establishes the proof of correctness.

Section \ref{sec-qpcc} is devoted to the Quantum Padded Coupon Collector problem, in order to prove the lower bound in \Cref{QPCC_thm}. In \Cref{sec-qpcc-ensemble}, we reduce the problem to finding a lower bound on the average error over a suitably chosen adversarial distribution, finding a diagonalized form for the ensemble average. In \Cref{sec-relation-classical}, we exhibit the emerging classical behaviour of this ensemble as the size of the padding increases, proving \Cref{QPCC_thm}. Finally, we explain how this result gives a separation between proper versus improper learning, i.e., we prove \Cref{thm-qpac-prop-vs-improp} in \Cref{sec-pac-reduction}.

In the appendix, in \Cref{sec-approx-error}, we bound the deviation from classical behaviour of the ensemble average in \Cref{sec-relation-classical}. \Cref{sec_clas_cc} presents known results about the Coupon Collector process for completeness; Section \ref{sec-relation-classical} calls for these results.

\paragraph{Acknowledgements.}
This research was supported in part by NSERC Canada. P.S.\ is also supported by a Mike and Ophelia Lazaridis Fellowship.

\section{Preliminaries}
\label{sec-preliminaries}

\paragraph{General notation.}

For simplicity we omit ceilings or floors when specifying some integral quantities such as the number of samples used by learning algorithms.

For a positive integer~$d$, we denote the set~$\{0,1,\ldots,d-1\}$ as~$\naturals_{d}$ and the set~$\{1,\ldots,d\}$ as~$[d]$. For a sequence~$\bb\in \naturals_p^n$, define~$\support(\bb)\eqdef\{i\in[n]:\bb_i\ne 0\}$, the set of coordinates of~$\bb$ with non-zero elements. For a sequence~$\bi \in [n]^t$, define~$\range(\bi) \eqdef \{ \bi_j : j \in [t] \} $, the set of elements of~$[n]$ occurring in the sequence~$\bi$. These definitions arise naturally by treating sequences as functions. For a function~$f$ with domain~$[n]$, and~$\bi \in [n]^t$, we denote the coordinate-wise application of~$f$ to~$\bi$ as~$f(\bi)$. The scalar product~$\bx \cdot \by$ for~$\bx, \by \in \naturals_p^t$ is defined as
\[
\bx \cdot \by \speqdef \sum_{i = 1}^t \bx_i \by_i \pmod{p} \enspace.
\]

\paragraph{Quantum Information notation.}

For a thorough introduction to the basics of quantum information, we refer the reader to the book by Watrous~\cite{W18-TQI}. We briefly review the notation that we use in this article. 

We consider only finite dimensional Hilbert spaces in this work, and denote them either by capital script letters like~$\cA$ and $\cB$, or directly as~$\complex^m$ for a positive integer~$m$.
A \emph{register\/} is a physical quantum system, and we denote it by a capital letter, like~$A$ or~$B$. A quantum register~$A$ is associated with a Hilbert space~$\cA$, and the state of the register is specified by a unit-trace positive semi-definite operator on~$\cA$. The state is called a \emph{quantum state\/}, or also as a \emph{density operator\/}. We denote quantum states by lower case Greek letters, like~$\rho$ and~$\sigma$.  We use notation such as~$\rho^A$ to indicate that register~$A$ is in state~$\rho$, and may omit the superscript when the register is clear from the context. We use \emph{ket\/} and \emph{bra\/} notation to denote unit vectors and their adjoints in a Hilbert space, respectively. A quantum state is \emph{pure\/} if it has rank one, i.e.,~$\rho = \density{\phi}$ where~$\ket{\phi}$ is a unit vector. For convenience we sometimes also use the ket notation for unormalised vectors, with an explicit mention.
For a linear operator~$M$ on some Hilbert space, $\trnorm{M} \eqdef \trace \sqrt{M^\adjoint M}$ denotes the trace norm of~$M$ (also called the~$\ell_1$ or Schatten~$1$ norm).

For a set~$X \subseteq [n]$, we define~$\ket{\psi_X} \in \complex^n$ as the uniform superposition over elements of the set~$X$, and~$\Pi_X$ as the orthogonal projection onto the subspace spanned by the elements of~$X$ in the same space. I.e., 
\begin{align*}
\ket{\psi_X} & \speqdef \frac{1}{\sqrt{ \size{X}}} \sum_{i\in X} \ket{i} \enspace, \qquad \text{and} \\
\Pi_X & \speqdef \sum_{i\in X} \density{i} \enspace. 
\end{align*}

\paragraph{Quantum learning theory.}

We refer the readers to the book~\cite{SB14-ML} for an introduction to machine learning theory and the survey~\cite{AdW17-survey} for an introduction to  quantum learning theory. Here,  we briefly review the notation related to the PAC model that we study in this article.

For some finite, non-empty domain~$X$, we refer to a Boolean function~$c: X \rightarrow \set{0,1}$ as a \emph{concept}. 
We may also think of a concept as a bit-string in~$\set{0,1}^n$ for~$n \eqdef \size{X}$ which lists the value~$c$ assigns to each element of~$X$. A \emph{concept class\/} is a subset~$\scC \subseteq \set{0,1}^X$ of Boolean functions. For a concept~$c$, we refer to~$c(x)$ as the \emph{label\/} of~$x \in X$, and the tuple~$(x,c(x))$ as 
a \emph{labeled example\/}. We say a concept class~$\scC$ is \emph{non-trivial\/} if it contains two distinct concepts~$c_1,c_2$ such that for some~$x_1, x_2 \in X$, we have~$c_1(x_1) = c_2(x_1)$ and~$c_1(x_2) \ne c_2(x_2)$. 

A crucial combinatorial quantity in learning Boolean functions is the \emph{VC dimension\/} of a concept class, introduced by Vapnik and Chervonenkis~\cite{VC71-VC-dim}. We say a set~$S \eqdef \set{s_1,\ldots,s_d} \subseteq X$ is \emph{shattered\/} by a concept class~$\scC$ if for every~$a \in \set{0,1}^{d}$, there exists a concept~$c \in \scC$ such that~$(c(s_1),\ldots,c(s_d)) = a$. The \emph{VC dimension\/} of~$\scC$, denoted as~$\VC(\scC)$, is the size of the largest set shattered by~$\scC$.  

\paragraph{PAC model.}

Consider a concept class~$\scC \subseteq \set{0,1}^X$. The PAC (\emph{probably approximately correct\/}) model for learning concepts was introduced---in the classical setting---by Valiant~\cite{val84-PAC}, and was extended to the quantum setting by Bshouty and Jackson~\cite{BJ99-DNF}. In the quantum PAC model, a learning algorithm is given a~\emph{quantum PAC example oracle\/}~$\QPEX(c,D)$ for an unknown concept~$c \in \scC$ and an unknown distribution~$D$ over~$X$. The oracle~$\QPEX(c,D)$ does not have any inputs. When invoked, it outputs a superposition of labeled examples of~$c$ with amplitudes given by the distribution~$D$, namely, the pure state
\begin{equation*}
    \sum_{x \in \set{0,1}^n} \sqrt{D(x)} \, \ket{x,c(x)} \enspace,
\end{equation*}
which we call a \emph{quantum sample\/}, or simply a \emph{sample\/}. Note that measuring a quantum sample in the computational basis gives us a labeled example distributed according to~$D$, i.e., a classical sample.
We say a Boolean function~$h$, commonly called a \emph{hypothesis\/}, is an~$\epsilon$-\emph{approximation\/} of~$c$ (or has error~$\epsilon$) with respect to distribution~$D$, if 
\begin{equation}
     \Pr_{x \sim D} \left[ h(x) \neq c(x) \right] \quad \leq \quad \epsilon   \enspace.
\end{equation}
Given access to the oracle~$\QPEX(c,D)$, the goal of a quantum PAC learner is to find a hypothesis~$h$ that is an~$\epsilon$-approximation of~$c$ with sufficiently high success probability. 
\begin{definition}
For~$\epsilon, \delta \in [0,1]$, we say that an algorithm~$\cA$ is an~$(\epsilon,\delta)$-\emph{PAC quantum learner\/} for the concept class~$\scC$ if for every~$c\in\scC$ and distribution~$D$, given access to~$\QPEX(c,D)$, with probability at least~$1-\delta$, the algorithm~$\cA$ outputs a hypothesis~$h\in \set{0,1}^X$ which is an~$\epsilon$-approximation of~$c$. We say that~$\cA$ is a \emph{proper\/} learner if it always outputs a hypothesis~$h \in \scC$, and say that it is \emph{improper\/} otherwise.
\end{definition}
The \emph{sample complexity\/} of a quantum learner~$\cA$ is the maximum number of times~$\cA$ invokes the oracle~$\QPEX(c,D)$ for any concept~$c \in \scC$ and any distribution~$D$ over~$X$. The~$(\epsilon,\delta)$-\emph{PAC quantum sample complexity\/} of a concept class~$\scC$ is the minimum sample complexity of a~$(\epsilon,\delta)$-PAC quantum learner for~$\scC$. Since we can readily derive classical samples from quantum ones, quantum learning algorithms are at least as efficient as classical ones in terms of sample complexity, as well as other measures such as time and space complexity.

\paragraph{Quantum Coupon Collector.}

Let~$n$ be an integer~$\ge 3$. For a positive integer~$k \in (1,n)$, let~$S$ be a~$k$-element subset of~$[n]$, and let~$\ket{\psi_S}$ denote the uniform superposition over the elements of~$S$:
\begin{equation*}
\label{eq-qsample}
 \ket{\psi_S} \quad \eqdef \quad \frac{1}{\sqrt k} \sum_{i \in S} \ket{i} \enspace.   
\end{equation*}
This is a quantum analogue of a uniformly random sample from~$S$, and we call this a \emph{quantum sample\/} of~$S$. For ease of notation, we define~$m \eqdef n - k$. 

In the Quantum Coupon Collector problem, we are given~$n,k$ and quantum samples of an arbitrary but fixed, unknown~$k$-element subset~$S$, and we would like to learn the subset using as few samples as possible. By ``learning the subset'', we mean that we would like to determine, with probability at least~$1 - \delta$ for some parameter~$\delta \in [0, 1)$, all the~$k$ elements of the set~$S$. We are interested in the quantum sample complexity of the problem, i.e., the least number of samples required by a quantum algorithm to learn the set with probability of error at most~$\delta$. Observe that by permuting the elements of~$[n]$ by a uniformly random permutation, we can show that the optimal worst-case error equals the optimal average-case error under the uniform distribution over the sets.

We may view the (Quantum) Coupon Collector problem as a problem in the PAC learning model as follows. The concept class~$\scC_{k,n}$ corresponding to the problem with parameters~$k,n$ is the set of characteristic vectors of all size~$k$ subsets of~$[n]$, i.e.,~$\scC_{k,n} \eqdef \set{c \in \set{0,1}^n : \size{c} = k}$. Labeled examples for a concept~$c \in \scC_{k,n}$ are always drawn from the uniform distribution over~$i \in [n]$ such that~$c(i) = 1$. Hence the label in such examples is superfluous, and the quantum sample considered above is equivalent to a sample in the quantum PAC model. Learning the unknown subset corresponds to \emph{properly\/} learning the concept with approximation error less than~$1/k$.

\paragraph{Padded Coupon Collector.}

The \emph{Padded Coupon Collector\/} problem is a variant of the Coupon Collector problem, defined as follows. Let~$n,k,p$ be positive integers with~$1 < k < n$, and~$l$ a non-negative integer. For an unknown set $S\subseteq [n]$ of size $k$ and unknown function~$f:[n]\to \N_p \,$, we are given independent random samples $(i,f(i))$ for $i$ chosen uniformly at random from $S$. Given~$\delta \in [0,1)$, the goal is to output, with probability at least~$1 - \delta$, a size~$k$ subset $S'$ such that $|S'\setminus S|\leq l$, using as few samples as possible.

We call~$f(i)$ as the $\textit{padding\/}$ of $i$, $p$ as the \textit{padding length\/}, and~$l$ as the number of \emph{mismatches\/}. In the classical case, this problem is as hard as the analogous Coupon Collector problem (in which we allow $l$ mismatches), as we can convert instances of either problem to those of the other, with the same parameters (i.e.,~$k,n,l$) by adding or removing the padding. This holds for any value of $p$.

The quantum analogue, the \textit{Quantum Padded Coupon Collector\/} problem, with parameters~$n,k,p,l$ as before is defined as follows. For an unknown set $S\subseteq [n]$ of size $k$ and unknown function~$f:[n]\to \N_p \,$, we are given copies of \emph{quantum samples\/} $\ket{\psi_{S,f}}$, where
\[
\ket{\psi_{S,f}} \speqdef \frac {1}{\sqrt{k}} \sum_{i\in S} \ket{i} \ket{f(i)} \enspace.
\]
Given~$\delta \in [0,1)$, the goal is again to output, with probability at least~$1 - \delta$, a size~$k$ subset $S'$ with at most~$l$ mismatches, i.e., with~$|S'\setminus S|\leq l$, using as few samples as possible.

As before, quantum samples reduce to random samples as in the classical problem, when measured in the computational basis. Unlike the classical case, it is not possible to convert individual quantum samples to those of the unpadded variant, i.e., the Quantum Coupon Collector problem, when~$p > 1$. Applying the unitary transform $\ket i \ket{f(i)} \to \ket i \ket 0$ requires knowledge of the entire function~$f$.

\paragraph{Tail bounds.}

To analyse the classical Coupon Collector process in Section \ref{sec_clas_cc}, we need the Hoeffding Bound for Hypergeometric series \cite{Hoeffding63-bounds}. If out of $n$ balls, $\alpha n$ balls are red and $(1-\alpha)n$ balls are blue for some~$\alpha \in (0,1)$, and $X$ is the random variable corresponding to the number of red balls picked when picking $t$ out of $n$ balls without replacement, then we have the following concentration bounds.
\begin{equation}\label{bound_hyp_geo}
    \Pr[X-\alpha t\geq \lambda]\leq \exp\left(-\frac {2\lambda^2}{t}\right)\quad ,\quad  \Pr[X-\alpha t\leq -\lambda]\leq \exp\left(-\frac {2\lambda^2}{t}\right).
\end{equation}

\section{Quantum Coupon Collector}
\label{sec-coupon-collector}

In this section, we present a more efficient learning algorithm for the Quantum Coupon Collector problem. The number of samples used by the algorithm matches the recently established sharp lower bound~\cite{BNS24-sample-complexity} exactly in the leading order term, for constant probability of error.

\subsection{The algorithm}
\label{sec-qcc-algorithm}

When~$3m \ln( \e m) > n$, we reduce the Quantum Coupon Collector problem to its classical version. I.e., we measure~$k\ln k +\Theta(k)$ quantum samples~$\ket{\psi_S}$ in the standard basis and output the set of all coupons observed. So, we need only focus on the ``small''~$m$ regime, i.e., when~$3m \ln( \e m) \le n \,$.

As in Ref.~\cite{ABCKRW20-quantum-coupon-collector}, we try to learn the complement~$\setcomp{S}$ of the set~$S$ when~$m$ is small, but we do this differently. We maintain a guess~$G$ for the set~$\setcomp{S}$, and use the quantum samples to improve our guess via suitable measurements. The measurements of a quantum sample may reveal that an element~$x$ has been misclassified  as being in~$\setcomp{S}$, in which case we remove it from~$G$. Alternatively, we may discover a new element~$x$ that is likely to be in~$\setcomp{S}$. In this case, we add~$x$ to~$G$. Sometimes, the measurements do not reveal either type of element, and we move to the next sample. We show that after sufficiently many such iterations, the guess~$G$ equals~$\setcomp{S}$ with constant probability close to~$1$. 

More formally, the algorithm keeps track of the following subsets of~$[n]$:
\begin{enumerate}

\item $G$: This is the current guess for the complement~$\setcomp{S}$.

\item $U$: This is a set that the algorithm currently believes contains~$S$. 
\end{enumerate}
Initially, $G$ is empty. Ideally the algorithm keeps adding elements from the complement~$\setcomp{S}$ to this set~$G$ until all the desired elements have been collected. The set~$U$ is initially set to~$[n]$, and is updated to~$[n] \setminus G$, whenever~$G$ is updated. We refer to elements~$i \not\in \setcomp{S}$ that have been incorrectly added to~$G$ as \emph{rogue\/} coupons. The set of rogue coupons at any time is precisely~$G \intersect S$.

Recall the notation~$\ket{\psi_X}$ and~$\Pi_X$ for the uniform superposition over elements of the set~$X \subseteq [n]$ and the orthogonal projection onto the subspace spanned by the elements of~$X$, respectively. Let~$M_0 \eqdef \Pi_G$ and~$M_1 \eqdef \id - M_0 \,$. The algorithm first measures a quantum sample~$\ket{\psi_S}$ according to~$(M_0, M_1)$. If the outcome is~$0$, the guess~$G$ necessarily has rogue coupons. We measure the residual state, viz.~$\ket{\psi_{G \intersect S}}$, to identify a rogue coupon and remove it from~$G$. If the outcome is~$1$, the algorithm further measures the residual state~$\ket{\psi_{U \intersect S}}$ according to~$(E_0, E_1)$, where~$E_0 \eqdef \density{\psi_U}$ and~$E_1 \eqdef \id - E_0$. If the outcome is~$1$, as we show later, the elements of~$\setcomp{S} \intersect U$ occur with higher (positive) amplitude in the resulting residual state. We therefore measure the residual state in the computational basis, and add the outcome to~$G$. Algorithm~\ref{alg-qcc} summarises all the steps of the algorithm.

\begin{algorithm}[ht]
\caption{Quantum-Coupon-Collector \label{alg-qcc}}

\SetKwInOut{Input}{Input}
\SetKwInOut{Output}{Output~}

\Input{positive integers~$n, k$ such that~$1 < k < n$, error parameter~$ \delta \in (0,1)$, and quantum samples~$\ket{\psi_S}$ for an unknown size~$k$ subset~$S \subset [n]$}
\Output{subset~$T \subset [n]$}

\BlankLine
$m \leftarrow n - k$ \;
\lIf{$3m \ln( \e m) > n$}{$\ell \leftarrow k \ln k + k \ln \tfrac{1}{\delta}$}
\lElse{$\ell \leftarrow k \ln m + k \ln \tfrac{\e}{\delta}$} 
Obtain~$\ell$ copies of the quantum sample, viz.,~$\ket{\psi_S}^{\tensor \ell}$ \;
\uIf{$3m \ln( \e m) > n$}{
  Measure each of the~$\ell$ quantum samples~$\ket{\psi_S}$ in the computational basis \;
  $T \leftarrow$ the set of all the outcomes observed \;
}
\uElse(\hspace{1in} \texttt{/* $3m \ln( \e m) \le n$ */}){ 
  $G \leftarrow \emptyset$, $U \leftarrow [n]$ \;
  \For{$t = 1$ to~$\ell$}{
    $M_0 \leftarrow \Pi_G \,$, $M_1 \leftarrow \id - M_0$ \; \label{step-1st-measurement}
    Measure the~$t$-th quantum sample~$\ket{\psi_S}$ according to~$(M_0, M_1)$ to get an outcome~$a \in \set{0,1}$ and corresponding residual state~$\ket{\xi_a}$ \; \label{step-1st-outcome}
    \uIf{outcome~$a = 0$}{
      Measure the residual state~$\ket{\xi_0}$ in the computational basis to get some outcome~$x$ \; \label{step-measure-rogue}
      $G \leftarrow G \setminus \set{x}$, $U \leftarrow U \union \set{x}$ \; \label{step-remove-elt}
    }
    \uElse(\hspace{1in} \texttt{/* outcome~$a = 1$ */}){ \label{step-1st-outcome-1}
      $E_0 \leftarrow \density{\psi_U}$, $E_1 \leftarrow \id - E_0$ \;
      Measure the residual state~$\ket{\xi_1}$ according to~$(E_0, E_1)$ to get an outcome~$b \in \set{0,1}$ and corresponding residual state~$\ket{\phi_b}$ \; \label{step-measure-u}
      \uIf{outcome~$b = 1$}{ 
        Measure the residual state~$\ket{\phi_1}$ in the computational basis to get some outcome~$x$ \; \label{step-measure-elt}
        $G \leftarrow G \union \set{x}$, $U \leftarrow U \setminus \set{x}$ \; \label{step-add-elt}
      }
    }
  }
  $T \leftarrow G$ \;
}
\Return{$T$}
\end{algorithm}

With ``high'' probability, in step~\ref{step-add-elt}, we keep adding elements of~$\setcomp{S}$ to~$G$. Due to our choice of~$U$ as the complement of~$G$ (and therefore a set disjoint from~$G$), the element added in this step is distinct from all the elements currently in~$G$. Steps~\ref{step-1st-measurement} to~\ref{step-remove-elt} are designed to correct our guess~$G$, since it is possible that in step~\ref{step-add-elt}, we mistakenly add elements from~$S$ to~$G$. By measuring~$\ket{\psi_S}$ with~$\ketbra{\psi_G}{\psi_G}$, we check if any rogue coupons were previously collected in~$G$. If rogue coupons are detected, we remove them from~$G$. If no rogue coupons were collected, or if all the ones that were collected have been removed, then steps~\ref{step-measure-rogue} and~\ref{step-remove-elt} are not executed. We claim that at the end of this procedure, $G$ is, with all but a small probability of failure, the set~$\setcomp{S}$. We prove the correctness of the algorithm in the next section.

\subsection{The analysis}
\label{sec-alg-analysis}

The correctness of Algorithm~\ref{alg-qcc} when~$3m \ln( \e m) > n$ follows from well-known properties of the classical Coupon Collector process. A straightforward calculation shows that the expected number of coupons that remain to be collected after~$t$ samples have been obtained is at most~$k( 1 - \tfrac{1}{k})^t$, which is bounded above by~$\delta$ when~$t = k \ln k + k \ln \tfrac{1}{\delta} \,$. This implies that the probability that the algorithm \emph{does not\/} collect all the~$k$ coupons in~$S$ with~$k \ln k + k \ln \tfrac{1}{\delta}$ samples is at most~$\delta$.

So we need only analyse Algorithm~\ref{alg-qcc} when~$3m \ln( \e m) \le n$. We do so by studying how far the guess~$G$ is from~$\setcomp{S}$, i.e., the size of the symmetric difference of the two sets. We show that in expectation, this distance is strictly decreasing, at a sufficiently high rate. To implement this approach, we keep track of two more sets.
\begin{enumerate}
\item $R$: the set of rogue coupons. This is the set $S\cap G$.
\item $C$: the set of coupons from $\setcomp{S}$ that remain to be collected. This is the set $\setcomp{S} \setminus G$
\end{enumerate}
We denote the cardinality of $R$ and $C$ by $J_t$ and $L_t$, respectively, after iteration~$t$. So~$J_0 = 0$ and~$L_0 = m$. Note that~$(J_t : t \ge 0)$ and~$(L_t : t \ge 0)$ are random walks over non-negative integers, with~$J_t \in [0,k]$ and~$L_t \in [0,m]$. 

The number of rogue coupons~$J_t$ may decrease in step~\ref{step-remove-elt}, or increase by~$1$ in step~\ref{step-add-elt}, or stay the same, in one iteration. The number of coupons~$L_t$ in~$C$ may stay the same, or decrease by~$1$ in step~\ref{step-add-elt}, in one iteration. Finally, at most one of~$J_t$ or~$L_t$ changes in one iteration. The algorithm succeeds if and only if there are neither any rogue coupons in~$G$ nor any coupons from~$\setcomp{S}$ that remain to be collected when it terminates, i.e.,~$J_\ell = L_\ell = 0$, where~$\ell$ is the total number of iterations. This is equivalent to the sum~$J_\ell + L_\ell$ being~$0$. So we track the random walk~$K_t \eqdef J_t + L_t$ instead.

Conditioned on~$J_t = j$ and~$L_t = l$ for some~$j \ge 0$ and~$l \in [0,m]$, we compute some states and the probability of some events occurring in the~$(t+1)$-th iteration. The probability of obtaining outcome~$0$ in step~\ref{step-1st-outcome}, and thus reducing~$J_t$ by~$1$, is 
\[
\norm{ \Pi_G \ket{\psi_S} }^2 \speq \sum_{i\in G\cap S} \frac{1}{k} \speq \frac{j}{k}\enspace.
\]
The probability of executing step~\ref{step-1st-outcome-1} is therefore~$1 - \frac{j}{k}$. Conditioned on executing step~\ref{step-1st-outcome-1}, i.e., getting outcome~$1$ in step~\ref{step-1st-outcome}, the residual state is~$\ket{\psi_{S\cap U}}$. So the conditional probability of getting outcome~$1$ in step~\ref{step-measure-u} is
\[
1 - \braket{\psi_U}{\psi_{S\cap U}}^2 \speq 1 - \frac{ |S\cap U| }{ |U|} 
    \speq 1-\frac {k-j}{k-j+l}
    \speq \frac{l}{k-j+l} \enspace,
\]
as~$S \cap U = S \setminus R$, and~$U = (S \cap U) \cup C$. The unnormalised resultant state on getting outcome~$1$ is
\begin{align*}
( \id - \ketbra{\psi_U}{\psi_U}) \ket{\psi_{S\cap U}} 
  & \speq \ket{\psi_{S\cap U}} - \braket{\psi_U}{\psi_{S\cap U}} \ket{\psi_U} \\
  & \speq \ket{\psi_{S\cap U}} - \sqrt{\frac{k-j}{k-j+l}} \; \ket{\psi_U} \\
  & \speq \ket{\psi_{S\cap U}} - \frac {k-j}{k-j+l} \; \ket {\psi_{S\cap U}} - \frac {\sqrt {(k-j)l}}{k-j+l} \; \ket{\psi_{C}} \\
  & \speq \frac{l}{k-j+l} \; \ket {\psi_{S\cap U}} - \frac {\sqrt {(k-j)l}}{k-j+l} \;\ket{\psi_{C}} \enspace.
\end{align*}
So the normalised state~$\ket{\phi_1}$ is
\[
\ket{\phi_1} \speq
    \sqrt{\frac {l}{k-j+l}} \; \ket{\psi_{S\cap U}} - \sqrt{\frac {k-j}{k-j+l}} \;  \ket{\psi_{C}} \enspace.
\]
From the above expression for~$\ket{\phi_1}$, it is immediate that the probability of adding a new rogue coupon to~$G$ in step~\ref{step-add-elt}, conditioned on getting this state, is~$\frac{l}{k-j+l} \,$. Similarly, the conditional probability for adding a fresh coupon from $\setcomp{S}$ to~$G$ is~$\frac {k-j}{k-j+l} \,$. The corresponding unconditional probabilities in iteration~$t+1$ (for~$J_t$ increasing or~$L_t$ decreasing by~$1$) are 
\[
\frac{k-j}{k} \times \frac {l}{k-j+l} \times \frac l{k-j+l} 
  \qquad \text{and} \qquad \frac {k-j}k \times \frac {l}{k-j+l}\times \frac {k-j}{k-j+l} \enspace,
\]
respectively.

Using these results, we can compute the transition probabilities for the random walk~$(K_t)$ in the~$(t+1)$-th iteration. We have~$K_{t+1} = K_t + 1$ when the number of rogue coupons~$J_t$ increases by~$1$. This happens with probability 
\[
\frac {(k-j)l^2}{k(k-j+l)^2} \enspace.
\]
We have~$K_{t+1} = K_t - 1$ when either the number of rogue coupons~$J_t$ decreases, or when we collect a new coupon from~$\setcomp{S}$, i.e., $L_t$ decreases. This happens with probability
\[
\frac{j}{k} +\frac{(k-j)^2 l}{k(k-j+l)^2} \enspace.
\]
Thus, defining~$r \eqdef j + l$, we get    
\begin{align*}
\expct \left[ \left. K_{t+1} \right| J_t = j, \, L_t = l \right]
  & \speq r + \frac{(k-j)l^2}{k(k-j+l)^2} - \frac{j}{k} - \frac{(k-j)^2l}{k(k-j+l)^2} \\
  & \speq r - \frac {(k-j)(k-j-l)l}{k(k-j+l)^2} - \frac{r-l}{k} \\
  & \speq r \left( 1-\frac{1}{k} \right) - \frac{l}{k} \left(\frac {(k-j)(k-j-l)}{(k-j+l)^2}- 1 \right) \\
  & \speq r \left( 1-\frac{1}{k} \right) + \frac{l^2}{k} \left( \frac{3(k-j)+l}{(k-j+l)^2} \right) \\ 
  & \sple r \left( 1-\frac{1}{k} \right) + \frac{3 l^2}{(k-j+l)k} \\
  & \speq r \left( 1-\frac{1}{k} \right) + \frac{3 r l^2}{k (k-j+l)(j+l)} \enspace.
\end{align*}
We bound the second term from above as follows. For~$l = 0$, it is bounded by~$0$. Otherwise, as a function of~$j \in [0,k]$, the expression~$(k-j+l)(j+l)$ is minimised at~$j \in \set{0,k}$, and the minimum is~$(k+l)l$. Moreover, as a function of~$l$ the expression~$l/(k + l)$ is maximised at~$l = m$. So
\[
\frac{l^2}{(k-j+l)(j+l)} \sple \frac{l^2}{l(k+l)} \sple \frac{m}{n} \enspace.
\]
Using this, we get
\begin{align*}
\expct \left[ \left. K_{t+1} \right| J_t = j_, \; L_t = l \right]
  & \sple r \left( 1 - \frac{1}{k} \right) + \frac{r}{k} \times \frac{3m}{n} \\
  & \speq r \left(1- \frac{1}{k} \left(1-\frac{3m}{n} \right)\right) \enspace.
\end{align*}
Since~$K_t = r = j + l$ when~$J_t = j, L_t = l$, we have
\begin{align}
\nonumber
\expct[K_{t+1}]
  & \sple \expct[K_t] \left(1 - \frac{1}{k} \left(1-\frac{3m}{n} \right) \right) \\
\nonumber
  & \sple \expct[K_0] \left(1-\frac{1}{k} \left(1-\frac{3m}{n} \right) \right)^{t + 1} \\
\label{eq-expct-bound}
  & \speq m \left(1-\frac{1}{k} \left(1-\frac{3m}{n} \right) \right)^{t + 1} \enspace. 
\end{align}
Note that the behaviour of the random variable~$K_t$ is very similar to that of the number of coupons that remain to be collected in the classical Coupon Collector process. 

Let~$\ell \eqdef k\ln m +ck$ for a positive parameter~$c$ to be specified later. As~$3m \ln(\e m) \le n \,$, by Eq.~\eqref{eq-expct-bound},
\begin{align*}
\expct[K_{\ell}] & \sple m \left(1-\frac{1}{k} \left(1-\frac{3m}{n} \right) \right)^{k\ln m +ck} \\
            & \sple \exp\left(\frac{3m\ln m - c(n - 3m)}{n} \right) & (\text{using } 1 + z \le \e^{z}) \\
            & \sple \exp( 1 - c ) \enspace.
\end{align*}
Thus, the probability that~$K_\ell$ is not~$0$ is at most~$\exp(1 - c)$. Taking~$c \eqdef \ln \tfrac{\e}{\delta}$, we get a probability of failure of at most~$\delta$. This proves \Cref{thm-qcc-optimal}.

\section{Quantum Padded Coupon Collector}
\label{sec-qpcc}

In this section, we establish strong bounds on the sample complexity of the Quantum Padded Coupon Collector problem, thereby proving \Cref{QPCC_thm}. Observe that at the heart of the problem we have a question about learning a subset from an ensemble of quantum states. In \Cref{sec-qpcc-ensemble}, we analyse this ensemble to bring it into diagonal form---a form that is more amenable to analysis. We then relate the problem to the \emph{classical\/} Coupon Collector process via an approximation of the ensemble in \Cref{sec-relation-classical}) to conclude the desired bound. Finally, in \Cref{sec-pac-reduction}, we show how this bound gives us a strong lower bound for proper PAC learning.  

\subsection{Simplified ensemble}
\label{sec-qpcc-ensemble}

Assume we have an algorithm for the problem with parameters~$n,k,l,p$ that uses~$t$ samples. We consider the average-case success probability of the algorithm when the underlying set-function pair~$S,f$ are chosen uniformly at random from their respective domains. We show that this probability is suitably small, if~$t$ is bounded as in \Cref{QPCC_thm}. For ease of reference, we call the version of the Quantum Padded Coupon Collector problem with uniformly random~$S,f$, with access to a total of~$t$ samples, as~$\cT_1$. 

The goal of the task~$\cT_1$ is to approximate the unknown set~$S$ regardless of the padding function~$f$. So we consider the quantum state~$\sigma_S$ corresponding to~$t$ quantum samples~$\ket{\psi_{S,f}}^{\tensor t}$ for uniformly random $f$. Treating~$f$ as an element of~$\N_p^n$, we have
\begin{equation} 
\label{ens_avg}
\sigma_S \speq  \frac 1 {p^n} \sum_{f\in \N_p^n} (\ketbra{\psi_{S,f}}{\psi_{S,f}})^{\otimes t} 
  \speq \frac 1 {p^nk^{t}} \sum_{\bi,\bj\in S^t} \ketbra{\bi}{\bj}\otimes\sum_{f\in \N_p^n}\ketbra{f(\bi)}{f(\bj)}  \enspace,
\end{equation}
where, for ease of notation, we have rearranged the registers in the final expression to collect the padding together.

Viewing~$\naturals_p^n$ as the additive group~$\integers_p^n$, we see that~$\sigma_S$ is invariant under the action induced by~$f \mapsto f + g$, for~$g \in \naturals_p^n$. Thus~$\sigma_S$ is block-diagonal in the same basis as this action, for every~$S$. 
\suppress{
The action is given by the operator
\begin{equation}
\label{eq-action}
\sum_{i \in [n]} \density{i} \tensor \sum_{x \in \naturals_p} \ketbra{x + g_i}{x}
\end{equation}
for every pair of registers containing an element~$j$ of set~$S$ and the corresponding padding~$f(j)$. Note that the operator is diagonalised by~$\id \tensor F$, where~$F$ is the Fourier transform over~$\integers_p$. 
}
With this in mind, consider~$\rho_S \eqdef (\id \tensor F^{\tensor t}) \, \sigma_S (\id \tensor (F^\adjoint)^{\tensor t})$ where~$F$ is the Fourier transform over~$\integers_p$:
\[
F \speqdef \frac{1}{\sqrt{p}} \sum_{j \in \naturals_p} \upomega^{ij} \ketbra{j}{i} \enspace,
\]
and~$\upomega \eqdef \exp \big(\frac{2\pi \complexi}{p} \big)$ is a primitive~$p$-th root of unity.
We have
\begin{align}
\nonumber
\rho_S& \speq \frac 1 {p^nk^{t}} \sum_{\bi,\bj\in S^t} \ketbra{\bi}{\bj}\otimes\sum_{f\in \N_p^n}F^{\otimes t}\ketbra{f(\bi)}{f(\bj)}{(F^*)}^{\otimes t}\\
\label{rho_S_1}
    & \speq  \frac 1 {p^{n+t}k^{t}} \sum_{\bi,\bj\in S^t} \ketbra{\bi}{\bj}\otimes \sum_{\bx,\by\in \N_p^t}  \left(\sum_{f\in \N_p^n}\upomega^{f(\bi)\cdot \bx-f(\bj)\cdot\by}\right)\ketbra{\bx}{\by} \enspace.
\end{align}
Now, we compute the value in brackets in Eq.~\eqref{rho_S_1}, i.e.,
\begin{equation}
\label{eq-ft-expr}
\sum_{f\in \N_p^n}\upomega^{f(\bi)\cdot \bx-f(\bj)\cdot\by} \enspace.
\end{equation}
The scalar product $f(\bi)\cdot \bx$ can be interpreted as a weighted sum of the coordinates of the vector $f(\bi)$, with weights given by $\bx$. Consider what happens if we increase one component of $f$ by $1$. More specifically, suppose we increase the value of $f(q)$ by $1$ for~$q \in [n]$, and keep the rest of~$f$ the same to get a new function $f'$. Then, $f'(x)-f(x)=\mathbbm 1[x=q]$, and we have 
\[
(f(\bi)\cdot \bx-f(\bj)\cdot \by)-(f'(\bi)\cdot \bx-f'(\bj)\cdot \by) 
  \speq  \sum_{r\in [t] \,:\, \bi_r=q} \bx_r - \sum_{r\in [t] \,:\, \bj_r=q} \by_r \enspace.
\]
If this quantity is $0\pmod p$, then changing $f(q)$ does not change $f(\bi)\cdot \bx-f(\bj)\cdot \by \pmod p$. Otherwise, as $f(q)$ cycles through $\N_p$, the expression $f(\bi)\cdot \bx-f(\bj)\cdot \by$ cycles through all multiples of the quantity $\pmod p$. Since~$\sum_{i=0}^{p-1} \upomega^{ij} = 0$ for any non-zero~$j \in \naturals_p$, we have that if 
\[
\sum_{r\in [t] \,:\, \bi_r=q} \bx_r - \sum_{r\in [t] \,:\, \bj_r=q} \by_r \speq 0 \pmod{p}
\]
for all $q$, then the sum in Eq.~\eqref{eq-ft-expr} is $p^n$, and otherwise it equals~$0$. This leads us to defining the \textit{modular signature\/} $\ms:[n]^t \times \N_p^t \to \N_p^n$ of a pair~$(\bi, \bx)$; for each~$q \in [n]$, the~$q$-th coordinate of the modular signature is given by
\begin{equation}
\label{eq-def-ms}
\ms (\bi,\bx)_q \speqdef  \sum_{r\in [t] \,:\, \bi_r=q}\bx_r  \pmod{p} \enspace.
\end{equation}
With this notion, we can evaluate the expression in Eq.~\eqref{eq-ft-expr} as
\[
\sum_{f\in \N_p^n}\upomega^{f(\bi)\cdot \bx-f(\bj)\cdot\by} \speq
\begin{cases}
    p^n &\text{if }\ms(\bi,\bx)=\ms(\bj,\by)\\
    0 &\text{otherwise.}
\end{cases}
\]
Combining this with Eq.~\eqref{rho_S_1}, we get
\begin{align}
\label{rho_s_ms}
\nonumber
\rho_S & \speq  \frac 1 {p^{t}k^{t}} \sum_{\substack{\bi,\bj\in S^t \,;\, \bx, \by \in \naturals_p^t \,:\, \\ \ms (\bi, \bx)=\mathrm{ms}(\bj,\bx)}} \ketbra{\bi}{\bj}\otimes  \ketbra{\bx}{\by}\\
    & \speq  \frac 1 {p^{t}k^{t}} \sum_{\bb \in \N_p^n} \ketbra {\phi_{S,\bb}}{\phi_{S,\bb}} \enspace,
\end{align}
where~$\ket{\phi_{S,\bb}}$ is the unnormalised state defined as
\[
\ket{\phi_{S,\bb}} \speqdef \sum_{ \bi\in S^t \,,\, \bx \in \naturals_p^t \,:\, \ms(\bi,\bx)=\bb} \ket \bi \ket \bx \enspace.
\]
Note that the states~$\ket{\phi_{S,\bb}}$ are orthogonal for different~$\bb$, and
\begin{align}
\label{eq-norm}
\sum_{\bb} \norm{\phi_{S,\bb}}^2 = k^t p^t \enspace. 
\end{align}
Note also that if~$\ket{\phi_{S,\bb}} \neq 0$, i.e., the~$\bb$ occurs as the modular signature of some pair~$\bi \in S^t, \bx \in \naturals_p^t \,$, then for any~$i \in [n]$, if we have~$\bb_i \neq 0$, then~$i \in S$. In other words, if~$\ket{\phi_{S,\bb}} \neq 0$, we have~$\support(\bb) \subseteq S$.

Since the ensemble of states~$(\rho_S) $ is related to~$(\sigma_S)$ by the Fourier transform, the task~$\cT_1$ is equivalent to approximating~$S$ given~$\rho_S$, for a uniformly random subset~$S$. 

\subsection{Relation to the classical case}
\label{sec-relation-classical}

Next, we show that the states~$\rho_S$ may be approximated closely by states~$\rho'_S$ which we define below, and the latter are easier to analyse. In fact, we show that estimating the set~$S$ given~$\rho'_S$ is closely related to the \emph{classical\/} Coupon Collector process.

Define $\rho'_S$ to be the following quantum state.
\[
\rho'_S  \speqdef  \frac 1 {p^{t}k^{t}} \sum_{\bb \in \N_p^n} \frac {\norm{\phi_{S,\bb}}^2}{\norm{\phi_{\bb}}^2}\ketbra {\phi_{\bb}}{\phi_{\bb}} \enspace,
\]
where 
\[
\ket{\phi_\bb} \speqdef \sum_{\bi\in \supp(\bb)^t ,\, \bx \in \naturals_p^t \,:\, \ms(\bi,\bx)= \bb} \ket \bi\ket \bx \enspace.
\]
I.e., we replace $\ket{\phi_{S,\bb}}$ in $\rho_S$ with the unormalised state $\ket{\phi_\bb}$ and normalise appropriately to construct~$\rho'_S \,$. Note that the states~$\ket{\phi_\bb}$ are also orthogonal to each other for different~$\bb$. Due to the properties of the states~$\ket{\phi_{S,\bb}}$ described in \Cref{sec-qpcc-ensemble}, we see that~$\rho'_S$ has support on~$\ket{\phi_\bb}$ only if~$\support(\bb) \subseteq S$. Moreover, for modular signature~$\bb$ and any two subsets~$S,S'$ of size~$k$ both of which contain~$\support(\bb)$, we have~$\norm{ \phi_{S,\bb}}^2 = \norm{ \phi_{S',\bb}}^2$, by symmetry. (We may permute~$[n]$ to map~$S$ to~$S'$ while preserving~$\support(\bb)$ and the number of basis elements~$\ket{\bi, \bx}$ occurring in the two states.)

Consider the distance between~$\rho_S$ and~$\rho'_S \,$.
\begin{align}
\nonumber
\trnorm{\rho_S - \rho'_S} 
  & \sple \frac{1}{p^{t}k^{t}} \sum_{\bb \in \N_p^n} \norm{\phi_{S,\bb}}^2 \trnorm{ \frac{ \density{ \phi_{S,\bb}}}{ \norm{\phi_{S,\bb}}^2 } - \frac{\density{\phi_{\bb}}}{ \norm{\phi_{\bb}}^2} } \\
\label{eq-distance}
  & \sple \max_{T \,:\, \support(\bb)\subseteq T} \trnorm{ \frac{ \density{ \phi_{T,\bb}}}{ \norm{\phi_{T,\bb}}^2 } - \frac{\density{\phi_{\bb}}}{ \norm{\phi_{\bb}}^2} } \enspace,
\end{align}
by Eq.~\eqref{eq-norm} and the relationship between~$\support(\bb)$ and the set~$S$ described above. In \Cref{p_bound} in \Cref{sec-approx-error} we show that this distance is bounded above by~$\tfrac{1}{2} \sqrt{k^t / p} \,$. The high-level reason this bound holds is that most of the terms~$\ket{\bi,\bx}$ occuring in $\ket{\phi_{S,b}}$ also occur in $\ket{\phi_\bb}$.

Call the task of estimating $S$ with at most~$l$ mismatches given $\rho'_S$ instead of $\rho_S$ as $\cT_2$. (As before, the $S$ is chosen uniformly at random.) Denote the maximum success probability over all algorithms for task~$\cT_i$ as~$\opt(\cT_i)$, for~$i \in \set{1,2}$. We have
\begin{align}
\nonumber
\opt(\cT_1) & \sple \opt(\cT_2) + \frac{1}{2} \max_{S \,:\, \supp(\bb)\subseteq S } \trnorm{\frac{\ketbra{\phi_{S,\bb}}{\phi_{S,\bb}}}{\norm{ \phi_{S,\bb} }^2}-\frac{\ketbra{\phi_{\bb}}{\phi_{\bb}}}{\norm{ \phi_{\bb} }^2}} \\
\label{error_term_p}
  & \sple \opt(\cT_2) + \sqrt{ \frac{k^t}{p} } & \text{(By \Cref{p_bound}).}
\end{align}
So we focus our attention on bounding $\opt(\cT_2)$.

Since $\rho'_S$ is diagonalised via the modular signature, i.e., by the (unnormalised) states~$\ket{ \phi_\bb}$, without loss in generality, we assume that an optimal learner~$\cA$ for task $\cT_2$ first measures the modular signature. If $\cA$ observes $\bb$ as the modular signature, the resultant state is $\ket{\phi_\bb}$, regardless of the unknown set $S$. Since $\cA$ optimizes the average error, and the \emph{a posteriori\/} probability of any size~$k$ subset~$S'$ containing~$\support(\bb)$ is the same, the optimal algorithm~$\cA$ outputs a uniformly random such subset~$S'$. This implies that as long as the support of~$\bb$ has size at least~$k - l$, the algorithm produces an estimate with at most~$l$ mismatches with~$S$.

It remains to bound the probability of observing $\bb$ with $|\supp(\bb)| \ge k - l$. The probability of observing a modular signature~$\bb$ with support size~$b$ for a given $S$ and $\rho'_S$ is
\begin{align*}
\frac {1}{p^tk^t}\sum_{\bb \,:\, |\supp(\bb)|=b}\norm{\phi_{S,\bb}}^2
  & \speq \frac{1}{p^tk^t}\sum_{\bb \,:\, |\supp(\bb)|=b} \Big\lvert\{ (\bi,\bx): \bi\in S^t, \bx \in \naturals_p^t \,, \ms(\bi,\bx)=\bb\} \Big\rvert \\
        & \speq  \frac{1}{p^tk^t} \Big\lvert\{ (\bi,\bx): \bi\in S^t, \bx \in \naturals_p^t \,, |\supp(\ms(\bi,\bx))|=b\} \Big\rvert \\
        & \speq \Pr_{\bI\sim S^t,\bX\sim \N_p^t} \big[ \,|\supp(\ms(\bI,\bX))|=b \big] \enspace,
\end{align*}
where~$\bI$ and~$\bX$ are drawn uniformly at random from~$S^t$ and~$\N_p^t \,$, respectively.
Thus, the probability that the modular signature $\bb$ observed has $|\supp(\bb)|\geq b$ is bounded as
\[
\Pr_{\bI\sim S^t,\bX\sim \N_p^t} \Big[ \,|\supp(\ms(\bI,\bX))|\geq b \Big] 
  \sple \Pr_{\bI\sim S^t} \big[ \, \size{\range(\bI)} \geq b \big] \enspace.
\]
Note that the bound is exactly the probability of collecting at least $b$ coupons from~$S$ in $t$ steps. This leads us to an approximation variant of the classical Coupon Collector problem. Define~$\cT_3$ as the task of estimating the set~$S$ by a size~$k$ subset~$S'$ with~$\size{S' \setminus S} \le l$ mismatches, given~$t$ independent samples from an unknown, uniformly random subset $S\subseteq[n]$ of size $k$. The reasoning above gives us the following lemma. 
\begin{lemma}
\label{QPCC_to_CC}
The success probability $\opt(\cT_2)$ is bounded above by the optimal probability of success for the task~$\cT_3$.
\end{lemma}
The Coupon Collector process is well-studied, and a bound on the success probability for task~$\cT_3$ is likely known. For completeness, we derive a bound on the probability in \Cref{clas_coup_coll} in \Cref{sec_clas_cc}.

We put all this together to prove \Cref{QPCC_thm}, which we restate here for convenience.
\thmqpccbound*
\begin{proof}
Increasing $n$ and decreasing the number of samples~$t$ only makes the problem harder, and thus it suffices to show that the success probability is strictly smaller than~$1 - \delta$ when $n \eqdef k + 5l$ and $t \eqdef t_0 - 1$. Taking~$l$ and~$p$ as in the theorem statement, 
by \Cref{clas_coup_coll}, the probability of estimating a uniformly random size~$k$ subset with at most~$l$ mismatches is~$1 - 2 \delta$. By \Cref{QPCC_to_CC} we get $\opt(\cT_2) < 1- 2\delta $. By our choice of~$p$ and Eq.~\eqref{error_term_p} we have $\opt(\cT_1) < 1-2\delta+\delta = 1-\delta$. Since $\cT_1$ is the task of solving the Quantum Padded Coupon Collector problem for a uniformly random subset (and uniformly random padding), the theorem follows.
\end{proof}

\subsection{Reduction to proper PAC learning}
\label{sec-pac-reduction}

Similar to the classical case, for each $d$ and $\epsilon$, we describe a concept class for which proper quantum PAC learning (with up to a small constant probability of error) has sample complexity $\Omega\left(\frac d {\epsilon}\log \frac 1 \epsilon+\frac 1 \epsilon \log \frac 1 \delta\right)$. The latter summand is obtained via a straightforward argument as in Refs.~\cite{AW18-optimal-sample,BNS24-sample-complexity}, so we focus on the first summand. 

Fix positive integer parameters~$k,p,d$. Let $n \eqdef k+d$.
We define the concept class $\scC_{n,k,p}$ over the domain $X \eqdef [n]\times \N_p$ as follows:
\[
\scC_{n,k,p} \speqdef \{g: \exists S\subseteq [n], |S|=k,  g(i,x)=\mathbbm 1[i\in S] \; \forall i\in [n], x\in \N_p \} \enspace.
\]
In the above, we refer to the concept corresponding to $S$ as $g_S$. We may verify that for~$k \ge d$, the VC dimension of the class~$\scC_{n,k,p}$ is exactly $d$. This is a ``padded'' version of the concept class~$\scC_{n,k}$ used to show a separation between proper and improper learning in the classical case. It is related to the Quantum Padded Coupon Collector problem in the following sense.
\begin{lemma}
Any proper $(\epsilon,\delta)$-PAC quantum learner $\cal A$ for $\scC_{n,k,p}$ also solves the Quantum Padded Coupon Collector problem with the same parameters $n,k,p$, and with~$l \eqdef \lfloor \epsilon k\rfloor $, with probability at least $1-\delta$, with the same number of samples.
\end{lemma}
\begin{proof}
Consider the single quantum sample $\ket{\psi_{S,f}}$ for the Quantum Padded Coupon Collector problem, with an additional qubit in the state $\ket 1$. The joint (pure) state is
\[ 
\ket{\psi_{S,f}}\ket 1 \speq \sum_{i\in S} \frac 1 {\sqrt S} \ket{i,f(i)}\ket 1 \enspace.
\]
Consider the distribution $D_{S,f}$ which is uniform over the the $k$ elements $(i,f(i))\in [n]\times \N_p$ given by $i\in S$, and the concept $g_S$. The state $\ket{\psi_{S,f}}\ket 1$ is a quantum sample for the concept $g_S$ with the distribution $D_{S,f}$. Hence, given samples~$\ket{\psi_{S,f}}$, we can construct the same number of quantum samples for~$g_S$ corresponding to the distribution $D_{S,f}$, and then feed them to the quantum learner $\cA$ for~$\scC_{n,k,p}$.

Since $\cal A$ is a proper learner, with probability at least $1-\delta$, it produces a concept $g_{S'}\in \scC_{n,k,p}$ which is an $\epsilon$ approximation of $g_S$ with respect to the distribution $D_{S,f}$. From this, we get
    \begin{alignat*}{2}
      & &\Pr_{(I,X)\sim D_{S,f}}[g_{S'}(I,X)\neq g_{S}(I,X)]\quad \leq&\quad  \epsilon\\
      &\implies\quad & \Pr_{I\sim S}[\mathbbm 1[I\in S']\neq 1]\quad \leq& \quad \epsilon\\
      &\implies\quad & |S\setminus S'|\quad \leq&\quad  \epsilon\cdot k\\
      &\implies\quad & |S' \setminus S|\quad \leq&\quad  \lfloor\epsilon\cdot k\rfloor \speq l \enspace. \\
    \end{alignat*}
Hence, with probability at least $1-\delta$, the output $S'$ has at most~$l$ mismatches with~$S$, and satisfies the requirements of the Quantum Padded Coupon Collector problem.
\end{proof}

From this, and the bound in \Cref{QPCC_thm}, the following is immediate.
\begin{corollary} 
For $d\geq 5$ and $k \eqdef \frac d {5\epsilon} \,$, the sample complexity of any quantum algorithm for proper $(\epsilon,\delta)$-PAC learning the concept class~$\scC_{n,k,p}$ is at least $t_0 \,$, where $p$ and $t_0$ are as in \Cref{QPCC_thm}. In particular, for failure probability~$\delta \le \frac 1 8 \,$ and padding length~$p \eqdef 64 k^{k\ln \frac {k+1}{10l+1}}$, the sample complexity is at least $\Omega\left( k  \log \frac {k+1} {l+1}\right)$, i.e., $\Omega(\frac d \epsilon \log \frac 1 \epsilon)$.
\end{corollary}
Note that if we fix $\delta \eqdef \frac 1 8 \,$, the lower bound holds for a fixed sufficiently large value of the padding length~$p$, and hence for a fixed concept class $\scC_{n,k,p}$ independent of $\delta$.
Combined with the general lower bound of $\Omega\left(\frac 1 \epsilon \log \frac 1 \delta\right)$ from \cite{AW18-optimal-sample,BNS24-sample-complexity}, and the general improper (classical) PAC learning algorithm due to Hanneke~\cite{Han16-opt-PAC} which has sample complexity $\Order \! \left(\frac d \epsilon +\frac 1 \epsilon \log \frac 1 \delta \right)$,  we get \Cref{thm-qpac-prop-vs-improp}.

\bibliographystyle{plain}
\bibliography{bibl}

\appendix

\section{Error due to approximation}
\label{sec-approx-error}

In \Cref{sec-relation-classical}, we approximate an ensemble arising from the Quantum Padded Coupon Collector Problem with another ensemble that is easier to analyse. Here, we bound the error term that arises as a consequence; see Eq.~\eqref{error_term_p}.

We begin with two claims that are useful in bounding the error. The claims count the number of pairs of sequences~$(\bi,\bx)$ with a given modular signature under different conditions on~$\bi$. (See Eq.~\eqref{eq-def-ms} for the definition of modular signature.)
\begin{lemma}
\label{claim_6.2}
For any integers~$n, t, p \ge 1$, and sequence~$\bb \in \naturals_p^n \,$,
\begin{align*}
\lefteqn{ \Big| \set{ (\bi,\bx) : \bi \in \support(\bb)^t, \bx \in \naturals_p^t \,, \ms(\bi,\bx)=\bb } \Big| } \\
  &  \speq p^{t-|\support(\bb)|} \times \Big| \set{ \bi \in \support(\bb)^t : \range(\bi)=\support(\bb)  } \Big| \enspace.
\end{align*}
\end{lemma}
\begin{proof}
Consider~$(\bi, \bx)$ such that~$\bi \in \support(\bb)^t$ and~$\ms(\bi, \bx)=\bb$. First, we have~$\range(\bi) \subseteq  \support(\bb)$. Second, if there is some element~$q \in [n]$ that does not appear in the sequence~$\bi$, then $\bb_q=0$, i.e, $q \notin \support(\bb)$. So we also have~$\range(\bi) \supseteq  \support(\bb)$, and hence~$\range(\bi) = \support(\bb)$.

For a fixed $\bb$ and $\bi\in \support(\bb)^t$ with $\range(\bi) = \support(\bb)$, the equation $\ms(\bi, \bx)=\bb$ is a linear equation in the variables~$\bx$. Furthermore, for each~$j \in [n]$, $\bb_j$ is a sum of the variables in~$X_j \eqdef \set{ \bx_r : \bi_r = j}$. The sets~$(X_j)$ form a partition of the~$t$ variables in~$\bx$. Let~$t_j \eqdef \size{X_j}$. Note that the equation with~$\bb_j$ is non-trivial if and only if~$j$ appears in $\bi$, i.e., $j \in \range(\bi)$ and~$t_j > 0$. If the equation is non-trivial, it has~$p^{t_j-1}$ solutions. Thus, the total number of solutions~$\bx$ to $\ms(\bi,\bx)=\bb$ is
\[
p^{\sum_{j = 1}^n t_j - \size{ \range(\bi)} } \speq p^{t- \size{\range(i)}} \speq p^{t- \size{\support(\bb)} } \enspace.
\]
Multiplying this with the number of $\bi$'s with $\range(\bi)=\support(\bb)$, we get the claimed identity.
\end{proof}

\begin{lemma}
\label{claim_6.3}
For any integers~$n, t, p \ge 1$, set~$S \subseteq [n]$, and sequence~$\bb \in \naturals_p^n \,$,
\[
\Big|\set{(\bi,\bx): \bi\in S^t, \bx \in \naturals_p^t \,, \ms(\bi,\bx)=\bb } \Big| \speq  \sum_{T \,:\, \support(\bb)\subseteq T\subseteq S}p^{t-|T|} \times \Big|\set{\bi\in T^t : \range(\bi)=T  } \Big| .
\]
\end{lemma}
\begin{proof}
The proof is similar to that of the \Cref{claim_6.2}. We count the total number of pairs $(\bi,\bx)$ in the set on the LHS. Fix an~$\bi \in S^t$. Call $\range(\bi)$ as~$T$. By the same reasoning as in the proof of \Cref{claim_6.2}, since $\ms(\bi,\bx)=\bb$ for some $\bx$, we have~$\support(\bb)\subseteq T$. Moreover, there are $p^{t-|T|}$ sequences~$\bx$ with~$\ms(\bi,\bx)=\bb$. Thus, summing over the possible values for~$T$, we get the desired identity.
\end{proof}

Recall the unnormalised states~$\ket{\phi_{S,\bb}}$ and~$\ket{\phi_\bb}$ defined in \Cref{sec-qpcc}, where~$S \subset [n]$ and~$\bb \in \naturals_p^n$ is a modular signature of length~$t$ sequences. We show that the normalised states are close to each other in trace distance, provided the padding length~$p$ is sufficiently large.
\begin{lemma}
\label{p_bound}
For $S \subset [n]$ and $\bb \in \naturals_p^n$ such that $\support(\bb)\subseteq S$, 
\[
\norm{\frac{\ketbra{\phi_\bb}{\phi_\bb}}{\norm{\phi_\bb}^2}-\frac{\ketbra{\phi_{S,\bb}}{\phi_{S,\bb}}}{\norm{\phi_{S,\bb}}^2}}_1 
  \sple 2 \sqrt{\frac {k^t}{p}} \enspace,
\]
where $k=|S|$, $t$ is the number of samples, and $p$ is the padding length.
\end{lemma}
\begin{proof}
For any two pure states~$\ket{\xi}, \ket{\zeta}$ in the same space, we have
\[
\trnorm{ \density{\xi} - \density{\zeta} } \speq 2 \sqrt{1 - \size{\braket{\xi}{\zeta}}^2 \,} \enspace.
\]
So,  
\[
\norm{\frac{\ketbra{\phi_\bb}{\phi_\bb}}{\norm{\phi_\bb}^2} - \frac{\ketbra{\phi_{S,\bb}}{\phi_{S,\bb}}}{\norm{ \phi_{S,\bb} }^2}}_1 
  \speq 2 \sqrt {1-\left\lvert\frac {\braket{\phi_\bb}{\phi_{S,\bb}}} {\norm{ \phi_{S,\bb} }\norm{ \phi_\bb} } \right\rvert^2   } \enspace.
\]
Since the amplitudes of the standard basis elements in the states are real, we have
\[
\left\lvert\frac {\braket{\phi_\bb}{\phi_{S,\bb}}} {\norm{ \phi_{S,\bb}}\norm{ \phi_{\bb}}}\right\rvert^2 
  \speq \frac {\braket{\phi_\bb}{\phi_{S,\bb}}^2} {\norm{ \phi_{S,\bb}}^2\norm{ \phi_{\bb}}^2}
  \speq \frac {\braket{\phi_\bb}{\phi_\bb}^2} {\norm{ \phi_{S,\bb}}^2\norm{ \phi_{\bb}}^2}
  \speq \frac {\norm{ \phi_{\bb}}^2}{\norm{ \phi_{S,\bb}}^2} \enspace.
\]
Define $n_{t,T} \eqdef \size{ \set{\bi\in T^t:\range(\bi)=T } }$. 
By the definition of the states and \Cref{claim_6.2,claim_6.3} we have
\begin{align*}
\frac {\norm{ \phi_{\bb}}^2}{\norm{ \phi_{S,\bb}}^2}
    & \speq  \frac{ \size{ \set{(\bi,\bx): \bi \in \support(\bb)^t, \bx \in \naturals_p^t \,, \ms(\bi,\bx)=\bb }} }{ \size{ \set{ (\bi,\bx): \bi \in S^t, \bx \in \naturals_p^t \,, \ms(\bi,\bx)=\bb, }} } \\
    & \speq  \frac{p^{t-|\support(\bb) |} \, n_{t,\support(\bb)}  }{\displaystyle \sum_{T \,:\, \support(\bb)\subseteq T\subseteq S} p^{t-|T|} \, n_{t,T}  }\\
    & \speq \frac{n_{t,\support(\bb)} }{\displaystyle \sum_{T \,:\, \support(\bb)\subseteq T\subseteq S}p^{|\support(\bb)|-|T|}\, n_{t,T} }\\
    & \speq \frac{n_{t,\support(\bb)}  }{\displaystyle n_{t,\support(\bb)}+\sum_{T \,:\,\support(\bb)\subsetneq T\subseteq S}p^{|\support(\bb)|-|T|}\, n_{t,T}} \\
    & \spge \frac {1}{1+\frac {k^t}{p}} \enspace.
\end{align*}
Thus, the $\ell_1$ distance is at most $2\sqrt {1- \frac {1}{1+\frac {k^t}p}} \leq 2\sqrt {\frac {k^t} p}$. This proves the lemma.
\end{proof}

\section{Properties of the Classical Coupon Collector process}
\label{sec_clas_cc}

In this section, we derive some properties of the Coupon Collector problem for completeness. These are used in the classical case to show a separation between proper and improper PAC learning, and also turn out to be relevant to the quantum case. 

The parameters in the proof of \Cref{lem_cc_coll_coupons} below are based on the argument in Ref.~\cite{cc_props_ref}.
\begin{lemma}
\label{lem_cc_coll_coupons}
Let~$\delta \in [0,1)$ and~$l$ be a non-negative integer.
Given independent, uniform samples from a set~$S \subset [n]$ of size~$k$, at least~$k \ln \frac {k+1}{l+1} + k\ln (1-\delta)$ samples are required to observe at least~$k-l$ distinct elements of~$S$ with probability at least~$1-\delta$.
\end{lemma}
\begin{proof}
Suppose we draw samples from~$S$ until~$k-l$ distinct elements are observed. Then, the total number of samples is the sum of~$k-l$ independent geometric random variables. Let~$X_i$ be the number of samples drawn after the~$(i-1)$-th distinct element is observed until the $i$-th distinct element is observed (with the last sample included). Then~$X_i$ has a geometric distribution with parameter~$\frac{k-i+1}{k}$.

For any~$s, \lambda \ge 0$ we have
\begin{align*}
\Pr \left[ \sum_{i=1}^{k-l} X_i \le s \right] 
  & \speq \Pr \left[ \e^{-\lambda s} \leq \prod_{i=1}^{k-l} \e^{-\lambda X_i} \right]  \\
  & \sple \e^{\lambda s} \prod_{i=1}^{k-l} \mathbb E \left[ \e^{-\lambda X_i} \right] & \text{(By the Markov Inequality)} \\
  & \speq \e^{\lambda s} \prod_{i=1}^{k-l} \frac{\frac {k-i+1}k \e^{-\lambda} }{1-\frac{i-1}k \e^{-\lambda}}\\
  & \speq \e^{\lambda s}\prod_{i=1}^{k-l}\frac{\frac {k-i+1}k }{ \e^{\lambda}-\frac{i-1}k} \enspace.
\end{align*}
Let~$s \eqdef k\ln \frac{k+1}{l+1}+ck$ for some $c\in \mathbb R$.
Set $\lambda \eqdef \frac 1 k$ and notice that $\exp \left(\frac s k \right) = \frac {k+1}{l+1} \e^{c}$, and $\exp \left(\frac 1 k \right) \geq 1+\frac 1 k \,$. Using these, we get that
\begin{align*}
\Pr \left[ \sum_{i=1}^{k-l} X_i \le s \right] 
  & \sple \e^{c} \frac {k+1}{l+1} \prod_{i=1}^{k-l}\frac{\frac {k-i+1}k }{1+\frac{1}{k} -\frac{i-1}{k}} \\
  & \sple \e^{c} \frac {k+1}{l+1} \prod_{i=1}^{k-l}\frac{{k-i+1} }{k-i+2} \\
  & \speq \e^c \enspace.
\end{align*}
If~$c < \ln (1 - \delta)$, the probability that~$k-l$ distinct elements are observed with~$s$ samples is~$< 1 - \delta$. The lower bound stated in the lemma follows.
\end{proof}

The following lemma bounds the probability of correctly guessing most of the elements of an unknown set, when asked to guess a set of the same size.
\begin{lemma}
\label{lem-guess-set}
Let~$l,m$ be positive integers such that~$l \ge 10m$.
Suppose~$T \subseteq [l+5m]$ is an unknown subset of size~$l$. Then, if~$T' \subset [l+5m]$ is a uniformly random subset of size~$l$, the probability that it contains at least~$l-m$ elements of $T$ is at most~$\frac 1 2$.
\end{lemma}
\begin{proof}
For fixed~$m$, the probability in question decreases as~$l$ increases, so we prove the result for~$l \eqdef 10m$.

For~$i \in [l]$, let $X_i$ be the random variable indicating whether the $i$-th element of~$S$ (say, in increasing order) is contained in~$T'$. Note that~$\sum X_i$ is the number of elements of~$T$ contained in~$T'$. We have $\mathbb E[X_i] = \frac{l}{l+5m}=\frac 2 3$, and so $\mathbb E \left[ \sum_i X_i \right] =\frac {l^2}{l+5m}=\frac{20}3 m$. By the Hoeffding bound for the hypergeometric series (see Eq.~\eqref{bound_hyp_geo}),
\begin{align*}
\Pr \left[ \sum X_i\geq l-m \right] 
  & \speq \Pr\left[\sum X_i-\frac {20} 3 m\geq \frac {7} {3} m\right]\\
  & \sple \exp\left(-\frac {98}{90} m \right) \enspace.  
\end{align*}
This is at most~$\frac 1 \e \leq \frac 1 2 $.
\end{proof}

\Cref{lem-guess-set} helps bound the probability of properly learning the unknown set in the Coupon Collector problem.
\begin{lemma}
\label{clas_coup_coll}
Let~$\delta \in [0,1)$.
Given independent, uniformly random samples from an unknown set $S\subset [k+5m]$ with~$|S|=k$, at least $k \ln \frac {k+1}{10m+1}+k \ln(1-2\delta) $ samples are required in order to learn the set $S$ with at most $m\geq 1$ mismatches, with average probability of failure at most~$\delta$ when~$S$ is chosen uniformly at random.
\end{lemma}
\begin{proof}
By the Lemma \ref{lem_cc_coll_coupons}, we know that if the number of samples drawn is strictly smaller than~$k \ln \frac {k+1}{10m+1}+k\ln (1 - 2\delta)$, then there is at least $2\delta$ probability of observing at most $k-10m$ distinct elements from the set $S$.

Suppose~$k-l$ distinct elements were observed for some $l\geq 10m$, since we are interested in decreasing the average probability of failure for uniformly random~$S$, the best strategy is to guess the remaining $l$ elements uniformly at random. By \Cref{lem-guess-set}, the probability that we produce a guess of size~$k$ with up to~$m$ mismatches is at most $\frac 1 2$. Hence, with probability at least $2\delta \times \frac 12$, we fail to learn~$S$ up to $m$ mismatches.
\end{proof}

\end{document}